%% file: main.tex
\documentclass[runningheads,a4paper]{llncs}

\input{macros}

\title{Unambiguous B\"uchi is weak}

\titlerunning{Unambiguous B\"uchi is weak}

\author{Henryk Michalewski
\and Micha{\l} Skrzypczak}

\authorrunning{H. Michalewski and M. Skrzypczak}

\institute{University of Warsaw\\
Banacha 2 Warsaw, Poland\\
\mailsa}

\begin{document}

\maketitle

\begin{abstract}
A non-deterministic automaton on infinite trees is \emph{unambiguous} if it has at most one accepting run on every tree. 
For a  given unambiguous parity automaton $\Aa$ of index $(i,2j)$ we construct an alternating automaton $\textsc{Transformation}(\Aa)$  which accepts the same language, but is simpler in terms of alternating hierarchy of automata. If  $\Aa$ is a  B\"uchi automaton ($i\!=\!0,j\!=\!1$), then $\textsc{Transformation}(\Aa)$ is a weak alternating automaton. In general, $\textsc{Transformation}(\Aa)$  belongs to the class $\comp(\rmin+1,2\rmax)$, in particular it is simultaneously of alternating index $(i,2j)$ and of the dual index $(i+1,2j+1)$. The main theorem of this paper is a correctness proof of the algorithm $\textsc{Transformation}$. 
The transformation algorithm is based on a separation algorithm of Arnold and Santocanale~\cite{arnold_separation} and extends results of Finkel and Simonnet~\cite{finkel_borel}. 
\end{abstract}


\input{0_intro}

\input{1_basic}

\input{2_partition}

\input{3_construction}

\input{4_correctness}
\input{5_conclusion}


\bibliographystyle{plain}
\bibliography{mskrzypczak}

\end{document}

%% file: macros.tex
\usepackage{amssymb}
\usepackage{graphicx} 
\usepackage{microtype}
\usepackage[utf8]{inputenc}
\usepackage{todonotes}
\usepackage{url}
\usepackage{float}
\urldef{\mailsa}\path|{h.michalewski, m.skrzypczak}@mimuw.edu.pl|

\usepackage[pdfencoding=auto,]{hyperref}
\hypersetup{hidelinks}
\usepackage{bookmark}

\usepackage{wrapfig}

\usepackage {amsfonts}
\usepackage {amssymb,amsmath}
\usepackage {latexsym}
\usepackage {ifthen}
\usepackage {tikz}
\usepackage {xspace}
\usepackage {lmodern}
\usepackage {mathdots}
\usepackage[shortlabels]{enumitem}

\usepackage {url}
\usepackage {graphicx}
\usepackage {xcolor}
\usepackage {mathtools}
\usepackage[onelanguage]{algorithm2e}
\SetAlFnt{\footnotesize}

\SetCommentSty{mycommfont}

\newcommand{\trees}{Tr}
\newcommand{\dom}{dom}
\newcommand{\lang}{{\mathcal L}}

\usetikzlibrary{patterns, calc}

\usetikzlibrary{decorations.pathmorphing}

\tikzstyle{blob} = [
	decoration={random steps,
		segment length=0.8cm,
		amplitude=.3cm,
		pre=lineto,
		pre length=.45cm,
		post=lineto,
		post length=.45cm},
	decorate,
	rounded corners=.3cm]

\newcommand{\dashedLine}[4]{
\begin{scope}
\clip (0,0) -- (10,0) -- (10,10) -- (0, 10) -- cycle;
\begin{scope}[#1]
\draw[thick] (-20,0) -- (20,0);
\end{scope}
\end{scope}
}

\tikzstyle{lang} = [scale=1.2]

\tikzstyle{ubrace} = [draw, thick, decoration={brace, mirror, raise=0.1cm}, decorate,
    every node/.style={anchor=north, yshift=-0.2cm}]
\tikzstyle{rbrace} = [draw, thick, decoration={brace, mirror, raise=0.1cm}, decorate,
    every node/.style={anchor=west, xshift= 0.2cm}]

\tikzstyle{obrace} = [draw, thick, decoration={brace, raise=0.1cm}, decorate,
    every node/.style={anchor=south, yshift= 0.2cm}]
\tikzstyle{lbrace} = [draw, thick, decoration={brace, raise=0.1cm}, decorate,
    every node/.style={anchor=east, xshift=-0.2cm}]

\newcommand{\ignore}[1]{}
\numberwithin{equation}{section}

\newcommand{\ident}[2]{\newcommand{#1}{\ensuremath{\texorpdfstring{\mathrm{{#2}}}{#2}}\xspace}}

\newcommand{\domain}[2]{\newcommand{#1}{\ensuremath{\mathbb {#2}}}\xspace}

\newcommand{\ot}{\leftarrow}

\newcommand{\dL}{\mathtt{{\scriptstyle L}}}
\newcommand{\dR}{\mathtt{{\scriptstyle R}}}

\newcommand{\finA}{\mathnm{u}}
\newcommand{\finB}{\mathnm{v}}
\newcommand{\finC}{\mathnm{w}}

\newcommand{\infA}{\mathnm{\alpha}}
\newcommand{\infB}{\mathnm{\beta}}

\newcommand{\plyA}{\mathnm{\pi}}

\newcommand{\mysc}[1]{{\sc {\texorpdfstring{#1}{#1}}}\xspace}

\newcommand{\mso}{\mysc{mso}}

\domain{\N}{N}
\domain{\Z}{Z}
\domain{\Q}{Q}
\domain{\R}{R}
\domain{\Pri}{P}
\domain{\D}{D}

\newcommand{\aut}[1]{\ensuremath{\mathcal {#1}}\xspace}

\ident{\RM}{\mathbf{RM}}

\ident{\RMnd}{\mathbf{RM}^{\mathrm{nd}}}

\newcommand{\comp}{\mathrm{Comp}}

\ident{\cmp}{Comp}
\ident{\graph}{Graph}

\ident{\init}{I}

\newcommand{\eqdef}{\stackrel{\text{def}}=}

\newcommand{\mathcalsym}[1]{\ensuremath{\mathcal{#1}}\xspace}

\newcommand{\Aa}{\mathcalsym{A}}

\newcommand{\Cc}{\mathcalsym{C}}
\newcommand{\Dd}{\mathcalsym{D}}

\newcommand{\Gg}{\mathcalsym{G}}

\newcommand{\Ii}{\mathcalsym{I}}

\newcommand{\Rr}{\mathcalsym{R}}

\usetikzlibrary{positioning}
\usetikzlibrary{decorations.pathreplacing}
\usetikzlibrary{automata,positioning}
\usetikzlibrary{decorations.pathmorphing}
\usetikzlibrary{decorations.markings}
\usetikzlibrary{arrows}
\usetikzlibrary{patterns}
\usetikzlibrary{calc}

\tikzstyle{ubrace} = [draw, thick, decoration={brace, mirror, raise=0.0cm}, decorate,
    every node/.style={anchor=north, yshift=-0.1cm}]
    
\tikzstyle{snakearr} = [decorate, decoration={snake}, draw=blue, ->]
\tikzstyle{coord} = [draw,->,very thick]

\tikzstyle{edge}=[draw, thick,->]

\newcommand{\comment}[1]{}

\newcommand{\fun}[3]{\ensuremath{#1\colon #2 \to #3}}
\newcommand{\parfun}[3]{\ensuremath{#1\colon #2 \rightharpoonup #3}}

\ident{\rg}{rg}
\ident{\id}{id}
\ident{\sgn}{sgn}

\newcommand{\mathnm}[1]{{#1}}

\newcommand{\rmin}{\mathnm{i}}
\newcommand{\rmax}{\mathnm{j}}

\newcommand{\game}{\aut{G}}

\newcommand{\eve}{\ensuremath{\mathnm{\exists}}\xspace}

\newcommand{\adam}{\ensuremath{\mathnm{\forall}}\xspace}

\ident{\close}{cl}
\ident{\inter}{int}

\newcommand{\restr}{{\upharpoonright}}

%% file: 0_intro.tex
\section{Introduction}
%
 


Determinising a given computation typically leads to an additional cost. Presence of such cost inspires investigation of intermediate models of computations. Here we focus on {\em unambiguity}, that is the requirement that there are no two distinct accepting computations on the same input. 
In the case of finite and infinite words a given automaton can be determinised at an exponential cost, but in the case of infinite trees there are automata which cannot be determinised at all. Moreover, there are automata for which one cannot find an equivalent unambiguous automaton~\cite{niwinski_unambiguous}. Also, there exist unambiguous automata which cannot be simulated by deterministic ones~\cite{hummel_gandalf} (see Figure~\ref{fig:szczepan-example}).


Most questions about automata on finite or infinite words are decidable. 
However, in the case of automata on infinite trees many fundamental decidability problems are open, unless we limit attention to deterministic automata. 
Then it is decidable whether a given language is recognisable by a deterministic automaton~\cite{niwinski_deterministic}, the non-deterministic index problem is decidable~\cite{niwinski_relating,niwinski_gap}, as well as it is possible to locate the language in the Wadge hierarchy~\cite{murlak_wadge}. Moving beyond deterministic automata is a topic of an on-going research~\cite{colcombet_weak,murlak_game_auto} and the study of unambiguous automata is a part of this effort. Admittedly, problems for this class seem to be much harder than for deterministic automata, in particular one can decide if a given automaton is unambiguous, but it is an open problem, whether a given regular language is unambiguous. Additionally, there are no upper bounds on the descriptive complexity (e.g.~the parity index) or topological complexity of unambiguous languages among all regular tree languages.

In this work we focus on descriptive complexity and a fortiori also on topological complexity of languages defined by unambigous automata. 
The most canonical measure of descriptive complexity of regular tree languages is the \emph{parity index}. 
A parity automaton $\Aa$ has index $(\rmin,\rmax)$ if the priorities of the states of the automaton belong to the set $\{\rmin,\rmin+1,\ldots,\rmax\}$. In particular, the B\"uchi acceptance condition corresponds to the index $(1,2)$. It was shown in~\cite{arnold_strict,bradfield_simplifying} that some languages require big indices: for every pair $(\rmin,\rmax)$ there exists a regular language of infinite trees that is of index $(\rmin,\rmax)$ and cannot be recognised by any alternating nor non-deterministic automaton of a lower index. It means that the non-deterministic and alternating index hierarchies are \emph{strict}.

We will show that the fact that a given automaton is unambiguous allows to effectively find another equivalent automaton with a simpler acceptance condition. More precisely, in Section~\ref{sec:main-algorithm} we propose an algorithm $\textsc{Transformation}$ with the following properties: 

\newcommand{\thmmain}{
For an unambiguous B\"uchi automaton $\aut{A}$, $\textsc{Transformation}(\Aa)$ is a weak alternating automaton recognising the same language.  More generally, if $\aut{A}$ is an unambiguous automaton of index $(\rmin,2\rmax)$ then $\textsc{Transformation}(\Aa)$ accepts the same language as $\Aa$ and  belongs to the class $\comp(\rmin+1,2\rmax)$, in particular it is simultaneously of alternating index $(\rmin,2\rmax)$ and of the dual index $(\rmin+1,2\rmax+1)$.  

Additionally, the number of states of $\textsc{Transformation}(\Aa)$ is polynomial in the number of states of $\aut{A}$.
}

\begin{theorem}\label{th:unamb-Buchi-main}
\thmmain
\end{theorem}

This theorem implies in particular that there is no unambiguous B\"uchi automaton which is strictly of index $(1,2)$. Since a language accepted by an unambiguous B\"uchi automaton is also accepted by a weak alternating automaton, topologically such languages must be located at a finite level of the Borel hierarchy. 
One should note that in the above theorem and in the algorithm $\textsc{Transformation}$, the automaton must be simultaneously  unambiguous and of appropriate index. 
It is still possible for a regular tree language to be both: recognised by some unambiguous automaton and by some other B\"uchi automaton. An example of such a language is the $H$-language proposed in~\cite{hummel_gandalf}: ,,there exists a branch containing only $a$'s and turning infinitely many times right'', see Figure~\ref{fig:szczepan-example}.


\subsection{Related work}


There exist two estimates on descriptive complexity of unambiguous languages. Firstly, a result of Hummel~\cite{hummel_gandalf} shows that unambiguous languages are topologically harder than deterministic ones, see  Figure~\ref{fig:szczepan-example}. Secondly, Finkel and Simonnet~\cite{finkel_borel} proved using the Lusin-Souslin Theorem~\cite[Theorem~15.1]{kechris_descriptive}
from descriptive set theory, that any language recognised by an unambiguous B\"uchi automaton must be Borel.

Our theorem involves not only a set-theoretical argument but also an automata construction encapsulated by the algorithm \textsc{Transformation}. Our result also gives a stronger information about the descriptive complexity, since (1) it is an open problem whether for a given regular Borel language of infinite trees does exist a weak alternating automaton accepting this language, (2) our \textsc{Transformation} algorithm works for arbitrary parities and it is not clear how to generalize the set-theoretical method of Finkel and Simonnet~\cite{finkel_borel} beyond B\"uchi automata.

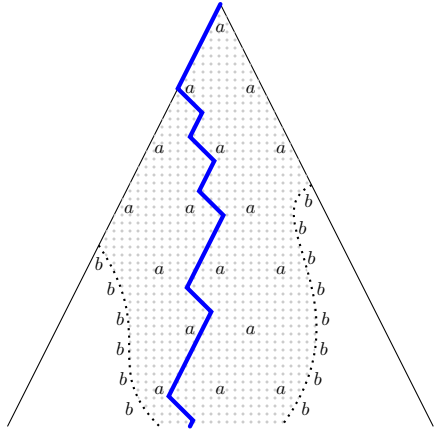
\begin{wrapfigure}[21]{l}{0.47\textwidth}
\centering
\begin{tikzpicture}[scale=0.8]

\draw (-3.5, -7) -- (0,0) -- (3.5,-7);

\path[pattern=dots, pattern color=black!20]
	(0,0) --
	(-2,-4) --
	(-2, -4) .. controls (-1.0,-5.5) and (-2.0,-6) .. (-1,-7) --
	(+1,-7)  .. controls (+2.5,-5) and (0.5,-3.5) ..  (1.5, -3) --
	cycle;

\draw[dotted, thick] (-2, -4) .. controls (-1.0,-5.5) and (-2.0,-6) .. (-1,-7);

\draw[dotted, thick] (1.5, -3) .. controls (0.5,-3.5) and (+2.5,-5) .. (+1,-7);

\foreach \x/\y in {
  0/0,
  -0.5/1,+0.5/1,
  -1.0/2, -0.0/2, +1.0/2,
  -1.5/3, -0.5/3, +0.5/3,
  -1.0/4, -0.0/4, +1.0/4,
  -0.5/5, +0.5/5,
  -1.0/6, -0.0/6, +1.0/6} {
  \node[scale=0.7, yshift=-13] at (\x,-\y) {$a$};
}

\foreach \x/\y in {
  -2/4.0,
  -1.8/4.4,
  -1.65/4.9,
  -1.65/5.4,
  -1.6/5.9,
  -1.5/6.4,
  1.45/2.95,
  1.35/3.4,
  1.5/3.9,
  1.65/4.4,
  1.73/4.9,
  1.7/5.4,
  1.6/5.9,
  1.4/6.4} {
  \node[scale=0.7, yshift=-10] at (\x,-\y) {$b$};
}

\path[draw=blue, ultra thick, cap=round] (0.0,0.0) --
	++(-0.7,-1.4) -- ++(0.4,-0.4) --
	++(-0.2, -0.4) -- ++(0.4,-0.4) --
	++(-0.25,-0.5) -- ++(0.4,-0.4) --
	++(-0.6,-1.2) -- ++(0.4,-0.4) --
	++(-0.7,-1.4) -- ++(0.4,-0.4) --
	++(-0.05,-0.1);

\end{tikzpicture}
\caption{A tree from the language $H$---the tree is labelled by letters $a$ and $b$, the dotted region contains vertices reachable from the root by $a$-vertices. The blue thick branch is a branch consisting of $a$-vertices that turns $\dR$ infinitely many times.}
\label{fig:szczepan-example}
\end{wrapfigure}

The Lusin-Souslin Theorem used in~\cite{finkel_borel}  says that if $\fun{f}{X}{Y}$ is injective and Borel then the image $f[X]$ is Borel in $Y$. The proof of this theorem is based on the Lusin Separation Theorem~\cite[Theorem~14.7]{kechris_descriptive}. These theorems are set-theoretical in nature and the result in this work can be considered as an automata-theoretic counterpart of the Lusin-Suslin theorem. As a sub-procedure in the algorithm \textsc{Transformation} we use an algorithm \textsc{Separation} from~\cite{arnold_separation}, which itself is an automata-theoretic counterpart of   
the Lusin Separation Theorem.


To the authors' best knowledge this is the first work where it is shown how to use the fact that a given automaton is unambiguous to derive upper bounds on the parity index of the recognised language. Therefore, this work should be treated as a first step towards descriptive complexity bounds for unambiguous languages, and generally better understanding of this class of automata.

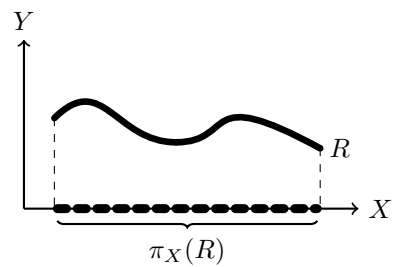
\begin{wrapfigure}{r}{0.45\textwidth}
\centering
\begin{tikzpicture}[scale=0.8]
\begin{scope}
\clip (-0.3, -1.0) -- (6.2,-1.0) -- (6.2, 3.4) -- (-0.3,3.4);
\path[draw,->,thick] (0,0) -- (0,2.8) node[anchor=south] {$Y$};
\path[draw,->,thick] (0,0) -- (5.5,0) node[anchor=west] {$X$};

\path[line width=2.5, draw=black, cap=round] (0.5,1.5) .. controls (1.3,2.3) and (1.5,1.1) .. (2.5,1.1);
\path[line width=2.5, draw=black, cap=round] (2.5,1.1) .. controls (3.5,1.1) and (4.87-2,2.1) .. (4.87,1.0);
\node[anchor=west] at (4.87,1.0) {$R$};

\path[draw=black, dashed] (0.5,1.5) -- (0.5, 0);
\path[draw=black, dashed] (4.87,1.0) -- (4.87, 0);
\path[draw=black, line width=3.4, loosely dotted, cap=round] (0.5+0.05,0) -- (4.87-0.05,0);
\path[ubrace] (0.5+0.05,-0.2) -- node {$\pi_X(R)$} (4.87-0.05,-0.2);
\end{scope}
\end{tikzpicture}
\caption{An illustration of Lusin-Souslin Theorem. A relation $R\subseteq X\times Y$ is Borel and \emph{uniformised}. The theorem implies that $\pi_X(R)\subseteq X$ is Borel as well.}
\label{fig:lusin-souslin}
\end{wrapfigure}

\subsection{Outline of the paper}
We first prove Lemma \ref{lem:unamb-buchi-disjoint} which states that if an automaton is unambiguous then the transitions of the automaton correspond to disjoint languages. In the algorithm \textsc{Partition} we use an algorithm of Arnold and Santocanale and show that these disjoint languages can be separated by $\comp(\rmin+1, 2\rmax)$ languages. 

In Section~\ref{sec:uBuchi-construction} we provide a construction of the automaton $\textsc{Transformation}(\Aa)$ and in Section \ref{sec:uBuchi-correctness} we conclude the proof of Theorem~\ref{th:unamb-Buchi-main} by proving correctness of this construction.

%% file: 1_basic.tex
\section{Basic notions}

In this section we introduce basic notions used in the rest of the paper. A good survey of the relations between deterministic, unambiguous, and non-deterministic automata is~\cite{colcombet_determinism}. A general background on automata and logic over infinite trees can be found in~\cite{thomas_languages}.

Our models are infinite, labelled, full binary trees. The labels come from a non-empty finite set $A$ called \emph{alphabet}. A tree $t$ is a function $\fun{t}{\{\dL,\dR\}^\ast}{A}$. The set of all such trees is $\trees_A$. Vertices of a tree are denoted $\finA,\finB,\finC\in\{\dL,\dR\}^\ast$. The prefix-order on vertices is ${\preceq}$, the minimal element of this order is the \emph{root} $\epsilon\in\{\dL,\dR\}^\ast$. The label of a tree $t\in\trees_A$ in a vertex $\finA\in\{\dL,\dR\}^\ast$ is $t(\finA)\in A$. $t\restr_\finA$. stands for the subtree of a tree $t$ rooted in a vertex $\finA$. Infinite branches of a tree are denoted as $\infA,\infB\in\{\dL,\dR\}^\omega$. We extend the prefix order to them, thus $\finA\prec \infA$ if $\finA$ is a prefix of $\infA$. For an infinite branch $\infA\in\{\dL,\dR\}^\omega$ and $k\in\omega$ by $\infA\restr_k$ we denote the prefix of $\infA$ of length $k$ (e.g.~$\infA\restr_0=\epsilon$).

A \emph{non-deterministic tree automaton} $\aut{A}$ is a tuple $\left<Q,A,q_0,\Delta,\Omega\right>$ where: $Q$ is a finite set of \emph{states}; $A$ is an alphabet; $q_\init\in Q$ is an \emph{initial state}; $\Delta\subseteq Q\times A\times Q\times Q$ is a \emph{transition relation}; $\fun{\Omega}{Q}{\N}$ is a \emph{priority function}.

If the automaton $\aut{A}$ is not known from the context we explicitly put it in the superscript, i.e.~$Q^\aut{A}$ is the set of states of $\aut{A}$.

A \emph{run} of an automaton $\aut{A}$ on a tree $t$ is a tree $\rho\in\trees_Q$ such that for every vertex $\finA$ we have $\big(\rho(\finA),t(\finA), \rho(\finA\dL),\rho(\finA\dR)\big)\in \Delta.$ A run $\rho$ is \emph{parity-accepting} if on every branch $\infA$ of the tree we have
\begin{equation}
\label{eq:parity-crit}
\limsup_{n\to\infty}\ \Omega\big(\rho(\infA\restr_n)\big) \equiv 0 \mod{2}. \tag{\color{blue}$\bigtriangleup$}
\end{equation}
We say that a run $\rho$ \emph{starts} from the state $\rho(\epsilon)$. A run $\rho$ is \emph{accepting} if it is parity-accepting and starts from $q_\init$. The \emph{language recognised} by $\Aa$ (denoted $\lang(\Aa)$) is the set of all trees $t$ such that there is an accepting run $\rho$ of $\aut{A}$ on $t$.

A non-deterministic automaton $\aut{A}$ is \emph{unambiguous} if for every tree $t$ there is at most one accepting run of $\Aa$ on $t$.

An \emph{alternating tree automaton} $\aut{C}$ is a tuple $\left<Q,A,Q_\eve, Q_\adam,q_0,\Delta,\Omega\right>$ where: $Q$ is a finite set of \emph{states}; $A$ is an alphabet; $Q_\eve\sqcup Q_\adam$ is a partition of $Q$ into sets of positions of the players \eve and \adam; $q_\init\in Q$ is an \emph{initial state}; $\Delta\subseteq Q\times A\times \{\epsilon, \dL,\dR\}\times Q$ is a \emph{transition relation}; $\fun{\Omega}{Q}{\N}$ is a \emph{priority function}. For technical reasons we assume that for every $q\in Q$ and $a\in A$ there is at least one transition $(q,a,d,q')\in \Delta$ for some $q'\in Q$ and $d\in\{\epsilon,\dL,\dR\}$.

An alternating tree automaton $\aut{C}$ induces, for every tree $t\in\trees_A$, a parity game $\game(\aut{C},t)$. The positions of this game are of the form $(\finA,q)\in \{\dL,\dR\}^\ast\times Q$. The initial position is $(\epsilon,q_\init)$. A position $(\finA,q)$ belongs to the player $\eve$ if $q\in Q_\eve$, otherwise $(\finA,q)$ belongs to $\adam$. The priority of a position $(\finA,q)$ is $\Omega(q)$. There is an edge between $(\finA,q)$ and $(\finA d,q')$ whenever $\left(q, t(\finA), d, q'\right)\in\delta$. An infinite play $\plyA$ in $\game(\aut{C},t)$ is winning for $\eve$ if the highest priority occurring infinitely often on $\plyA$ is even, as in condition~\eqref{eq:parity-crit}.

We say that an alternating tree automaton $\aut{C}$ \emph{accepts} a tree $t$ if the player $\eve$ has a winning strategy in $\game(\aut{C},t)$. The language of trees accepted by $\aut{C}$ is denoted by $\lang(\aut{C})$. A non-deterministic or alternating automaton $\aut{A}$ \emph{has index $(\rmin,\rmax)$} if the priorities of $\aut{A}$ are among $\{\rmin,\rmin+1,\ldots,\rmax\}$. An automaton of index $(1,2)$ is called a \emph{B\"uchi automaton}. Every alternating tree automaton can be naturally seen as a graph --- the set of nodes is $Q$ and there is an edge $(q,q')$ if $(q,a,d,q')\in \Delta$ for some $a\in A$ and $d\in\{\epsilon,\dL,\dR\}$. We say that an alternating tree automaton $\aut{D}$ is a \emph{$\comp(\rmin,\rmax)$ automaton} if every strongly-connected component of the graph of $\aut{D}$ is of index $(\rmin,\rmax)$ or $(\rmin+1,\rmax+1)$, see~\cite{arnold_separation}.

Note that an alternating automaton $\aut{C}$ is $\comp(0,0)$ if and only if $\aut{C}$ is a weak alternating automaton in the meaning of~\cite{kupferman_complementation}. The following fact gives a connection between these automata and weak~\mso (the variant of monadic second-order logic where set quantifiers are restricted to finite sets).

\begin{theorem}[Rabin~\cite{rabin_separation}, also Kupferman Vardi~\cite{kupferman_complementation}]
If $\aut{C}$ is an alternating $\comp(0,0)$ automaton then $\lang(\aut{C})$ is definable in weak~\mso. Similarly, if $L\subseteq\trees_A$ is definable in weak~\mso then there exists a $\comp(0,0)$ automaton recognising $L$.
\end{theorem}
The crucial technical tool in our proof is the \textsc{Separation} algorithm by Arnold and Santocanale~\cite{arnold_separation}. A particular case of this algorithm for $\rmin=\rmax=1$ is the classical Rabin separation construction (see~\cite{rabin_separation}): if $L_1$, $L_2$ are two disjoint languages recognisable by B\"uchi alternating tree automata then one can effectively construct a weak~\mso-definable language $L_S$ that separates them.

\RestyleAlgo{ruled}
\begin{algorithm}
\label{alg:nondet-sep}
\caption{\textsc{Separation}}

\SetAlgoVlined
\LinesNumbered

\KwIn{Two non-deterministic automata $\aut{A}_1$, $\aut{A}_2$ of index $(i,2\rmax)$ such that $\lang(\aut{A}_1)\cap \lang(\aut{A}_2)=\emptyset$.} 
\KwOut{An alternating $\comp(\rmin{+}1,2\rmax)$ automaton $\aut{S}$ such that
\[\lang(\aut{A}_1)\subseteq \lang(\aut{S})\quad\text{and}\quad\lang(\aut{A}_2)\cap\lang(\aut{S})=\emptyset.\]}
\end{algorithm}

%

%% file: 2_partition.tex
\section{Partition property}
\label{sec:partition}

In this section we will prove Lemma~\ref{lem:unamb-buchi-disjoint} stating that if an automaton $\Aa$ is unambiguous then the transitions of $\Aa$ need to induce disjoint languages. This will be important in the algorithm \textsc{Partition} which for a given unambiguous automaton of index $(\rmin,2\rmax)$, constructs a family of $\comp(\rmin+1,2\rmax)$ automata that split the set of all trees into disjoint sets corresponding to the respective transitions of $\Aa$. \textsc{Partition} will be used in \textsc{Transformation}.

Let us fix an unambiguous automaton $\aut{A}$ of index $(\rmin,2\rmax)$. Let $Q$ be the set of states of $\aut{A}$ and $A$ be its working alphabet. We say that a transition $\delta=(q,a, q_\dL,q_\dR)$ of $\aut{A}$ \emph{starts} from $(q,a)$; let $\Delta_{q,a}$ be the set of such transitions.

A pair $(q,a)\in Q\times A$ is \emph{productive} if it appears in some accepting run: there exists a tree $t\in \trees_A$ and an accepting run $\rho$ of $\aut{A}$ on $t$ such that for some vertex $\finA$ we have $\rho(\finA)=q$ and $t(\finA)=a$. This definition combines two requirements: that there exists an accepting run that leads to the pair $(q,a)$ and that some tree can be parity-accepted starting from $(q,a)$. Note that if $(q,a)$ is productive then there exists at least one transition starting from $(q,a)$. Without changing the language $\lang(\Aa)$ we can assume that if a pair is not productive then there is no transition starting from this pair.

For every transition $\delta=(q,a,q_\dL,q_\dR)$ of $\aut{A}$ we define $L_{\delta}$ as the language of trees such that there exists a run $\rho$ of $\aut{A}$ on $t$ that is parity-accepting and \emph{uses $\delta$ in the root of $t$} $\rho(\epsilon)=q$, $t(\epsilon)=a$, $\rho(\dL)=q_\dL$, and $\rho(\dR)=q_\dR$. Clearly the language $L_\delta$ can be recognised by an unambiguous automaton of index $(\rmin,2\rmax)$. If $(q,a)$ is not productive then $L_{(q,a,q_\dL,q_\dR)}=\emptyset$. The following lemma is a simple consequence of unambiguity of the given automaton $\aut{A}$.

\begin{lemma}
\label{lem:unamb-buchi-disjoint}
If $\delta_1\neq \delta_2$ are two transitions starting from the same pair $(q,a)$ then the languages $L_{\delta_1}$, $L_{\delta_2}$ are disjoint.
\end{lemma}

\begin{proof}
First, if $(q,a)$ is not productive then by our assumption $L_{\delta_1}=L_{\delta_2}=\emptyset$. Assume contrary that $(q,a)$ is productive and there exists a tree $r\in L_{\delta_1}\cap L_{\delta_2}$ with two respective parity-accepting runs $\rho_1$, $\rho_2$. Since $(q,a)$ is productive so there exists a tree $t$ and an accepting run $\rho$ on $t$ such that $\rho(\finA)=q$ and $t(\finA)=a$ for some vertex $\finA$. Consider the tree $t'=t[\finA\ot r]$ --- the tree obtained from $t$ by substituting $r$ as the subtree under $\finA$. Since $\rho(\finA)=q$ and both $\rho_1$, $\rho_2$ start from $(q,a)$, we can construct two accepting runs $\rho[\finA\ot \rho_1]$ and $\rho[\finA\ot \rho_2]$ on $t'$. Since these runs differ on the transition used in $\finA$, we obtain a contradiction to the fact that $\aut{A}$ is unambiguous. $\hfill\qed$
\end{proof}
The above lemma will be important in the algorithm \textsc{Partition}, because it uses the \textsc{Seperation} algorithm which in turn requires disjointness of the languages. 

\RestyleAlgo{ruled}
\begin{algorithm}
\label{alg:partition}
\caption{\textsc{Partition}}

\SetAlgoVlined
\LinesNumbered
\KwIn{An unambiguous automaton $\aut{A}$ of index $(\rmin,2\rmax)$} 
\KwOut{for every $\delta\in\Delta$ an automaton $\Cc_{\delta}$}

\ForEach{$(q,a)\in Q\times A$, productive}{

\ForEach{$\delta\in\Delta_{q,a}$}{
   $\aut{E}_{\delta}\leftarrow$ \emph{non-det. $(\rmin,2\rmax)$ automaton recognising $L_\delta$}
   
   $\aut{F}_{\delta}\leftarrow$ \emph{non-det. $(\rmin,2\rmax)$ automaton recognising $\bigcup_{\eta\in\Delta_{q,a}, \eta\neq\delta} L_{\eta}$}
}
\ForEach{$\delta\in\Delta_{q,a}$}{
   $\aut{D}_{\delta}\leftarrow \textsc{Separation}(E_\delta,F_\delta)$
}
\ForEach{$\delta\in\Delta_{q,a}$}{
   $\aut{C}_{\delta}\leftarrow$ $\comp(\rmin{+}1,2\rmax)$   \emph{automaton recognising} $\lang(\aut{D}_\delta)\setminus \bigcup_{\eta\neq\delta}\lang(\aut{D}_{\eta})$.
}
   $\aut{B}_{q,a}\leftarrow$ $\comp(\rmin{+}1,2\rmax)$  \emph{automaton recognising} $\trees_A\setminus\bigcup_{\delta\in\Delta_{q,a}} \lang(\aut{D}_{\delta})$.
}

\ForEach{$\delta=(q,a,q_\dL,q_\dR)\in\Delta_{q,a}$ with $(q,a)$ non-productive}{
   $\aut{C}_{\delta}\leftarrow$ $\comp(0,0)$   \emph{automaton recognising the empty language}.
}
\end{algorithm}

\noindent
The following lemma summarizes properties of the algorithm $\textsc{Partition}$.
\begin{lemma}
\label{lem:C-delta-properties}
Assume that $\Aa$ is an unambiguous automaton of index $(\rmin,2\rmax)$ and let $(q,a)\in Q\times A$. Take the automata $\big(\Cc_\delta\big)_{\delta\in\Delta_{q,a}}$ constructed by $\textsc{Parition}(\Aa)$. Then the languages $\lang(\Cc_\delta)$ for $\delta\in\Delta_{q,a}$ are pairwise disjoint and $L_\delta\subseteq \lang(\Cc_\delta)$.
\end{lemma}
A proof of this lemma follows directly from the definition of the respective automata, see Figure~\ref{fig:partition} for an illustration of this construction.
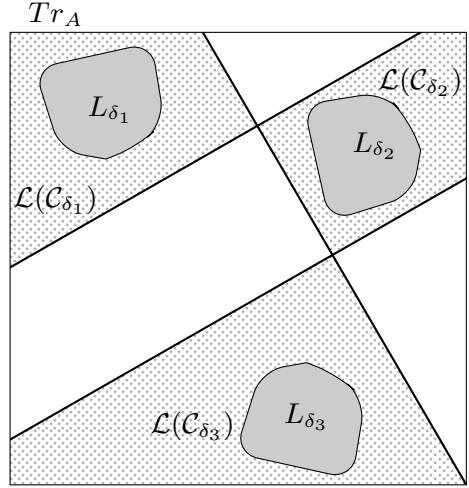
\begin{figure}
\centering
\begin{tikzpicture}[scale=0.6]

\coordinate (c1) at (5.40, 7.90);
\coordinate (c2) at (7.05, 5.10);

 \path[pattern=crosshatch dots, pattern color=black!30] 	(0, 4.8) -- (0, 10) -- (4.2, 10) -- (c1) -- cycle;

 \path[pattern=crosshatch dots, pattern color=black!30]	(0, 1) -- (c2) -- (10, 0) -- (0,0) -- cycle;

 \path[pattern=crosshatch dots, pattern color=black!30] 	(c1) -- (9,10) -- (10, 10) -- (10,6.8) -- (c2) -- cycle;

\node[lang] at (1, 10.4) {$\trees_A$};

\draw (0,0) -- (10,0) -- (10,10) -- (0, 10) -- cycle;

\draw[fill=black!20,blob] (2.1,7.2) arc (-90:270:1) ;

\draw[fill=black!20,blob] (9,7.4) arc ( 0:360:1.0) ;

\draw[fill=black!20,blob] (6.5,2.7) arc (90:-270:1) ;

 \dashedLine{shift={(0,1)},rotate=30}{0.3}{north west lines}{black!60}

 \dashedLine{shift={(10,0)},rotate=120}{0.3}{north east lines}{black!60}

 \dashedLine{shift={(9,10)},rotate=210}{0.3}{north west lines}{black!60}

\node[lang] at (2.2,8.3) {$L_{\delta_1}$};

\node[lang] at (8,7.5) {$L_{\delta_2}$};

\node[lang] at (6.5,1.5) {$L_{\delta_3}$};

\node[lang] at (1.0,6.3) {$\lang(\Cc_{\delta_1})$};

\node[lang] at (9.0,8.9) {$\lang(\Cc_{\delta_2})$};

\node[lang] at (4.0,1.3) {$\lang(\Cc_{\delta_3})$};

\end{tikzpicture}
\caption{An illustration of the output of the algorithm \textsc{Partition}. The three circles are the languages $L_{\delta_i}$ for the transitions starting in a fixed pair $(q,a)$.
Each straight line represents the language $\lang(\Dd_{\delta_i})$ that separates the respective language $L_{\delta_i}$ from the others.
Our construction provides the automata $\Cc_{\delta_i}$ recognising the dotted regions.}
\label{fig:partition}
\end{figure}

%% file: 3_construction.tex
\section{Construction of the automaton}
\label{sec:construction}
\label{sec:main-algorithm}
\label{sec:uBuchi-construction}

In this and the following section we will describe the algorithm \textsc{Transformation} and prove Theorem~\ref{th:unamb-Buchi-main} which states correctness and properties of  this algorithm. Given an automaton $\aut{A}$ of index $(\rmin,2\rmax)$,  the algorithm \textsc{Transformation} constructs an alternating $\comp(\rmin+1,2\rmax)$ automaton $\aut{R}$ recognising $\lang(\aut{A})$. 
It will consist of two sub-automata running in parallel:
\begin{enumerate}
\item In the first sub-automaton the role of \eve will be to propose a partial run $\parfun{\rho}{\{\dL,\dR\}^\ast}{Q}$ on a given tree $t$. She will be forced to propose certain unique run $\rho_t$ that depends only on the tree $t$, see Definition~\ref{def:unique}. At any moment \adam can \emph{challenge} the currently proposed transition and check if it agrees with the definition of $\rho_t$ (namely Condition~\eqref{eq:rho-t-constraint}).
\item In the second sub-automaton the role of \adam will be to prove that the partial run $\rho_t$ is not parity-accepting. That is, he will find a leaf in $\rho_t$ or an infinite branch of $\rho_t$ that does not satisfy the parity condition. Since the run $\rho_t$ is unique, \adam can declare in advance what will be the odd priority $n$ that is the limes superior (i.e.~$\limsup$) of priorities of $\rho_t$ on the selected branch.
\end{enumerate}

The automaton $\aut{R}$ consists of an \emph{initial component} $\Ii$ and of the union of the automata $\aut{C}_{\delta}$ constructed by the procedure $\textsc{Partition}$.

\RestyleAlgo{ruled}

\begin{algorithm}
\label{alg:transformation}
\caption{\textsc{Transformation}}

\SetAlgoVlined
\LinesNumbered
\KwIn{An unambiguous automaton $\aut{A}$ of index $(\rmin,2\rmax)$} 
\KwOut{An automaton $\Rr$}

$N \leftarrow \{\star\}\cup \big\{n\in\{\rmin,\ldots,2\rmax\}\mid \text{$n$ is odd}\big\}$

$Q_{\Ii,\eve}\leftarrow Q^\Aa\times N\sqcup\{\bot,\top\}$

$Q_{\Ii,\adam}\leftarrow \Delta^\Aa\times N$

$\Delta_{\Ii}\leftarrow \big\{(\bot,a,\epsilon,\bot),(\top,a,\epsilon,\top)\mid a\in A^{\Aa}\big\}$

$q_\init^\Rr\leftarrow (q_\init^\Aa,\star)$

$\big(\aut{D}_{\delta}\big)_{\delta\in\Delta}\leftarrow \textsc{Partition}(\Aa)$

\ForEach{$a\in A$, $q\in Q^{\Aa}$, $n\in N$}{
\uIf{$n\neq\star$ and $\Omega^\aut{A}(q)>n$}{
  $\Delta_{\Ii}\leftarrow\Delta_{\Ii}\cup \{((q,n),a,\epsilon,\top)\}$\label{it:adam_loses}
}
\Else{
  $\Delta_{\Ii}\leftarrow\Delta_{\Ii}\cup \Big\{\big((q,n),a,\epsilon,(\delta,n)\big)\mid\delta\in\Delta^\Aa_{q,a}\Big\}$\label{it:eve_trans}
}
}
\ForEach{$a\in A$, $\delta=(q,a,q_\dL,q_\dR)\in\Delta^\Aa$, $n\in N$}{
  $\Delta_{\Ii} \leftarrow \Delta_{\Ii}\cup \big\{(\delta, a,\epsilon, q_\init^{\Cc_\delta})\big\}\quad$\label{it:adam_challenge}  \tcc{ such a transition is a \emph{challenge}}
    
  \If{$n\neq \star$}{     
     $\Delta_{\Ii} \leftarrow \Delta_{\Ii}\cup \big\{(\delta, a,d, (q_d,n))\mid d\in\{\dL,\dR\} \big\}$\label{it:adam_proceed}
  }
  \Else{
     $\Delta_{\Ii} \leftarrow \Delta_{\Ii}\cup \big\{(\delta, a,d, (q_d,n'))\mid d\in\{\dL,\dR\}, n'\in N \big\}$\label{it:adam_choose}
  }
}

$Q^{\Rr}_\eve\leftarrow Q_{\Ii,\eve}\sqcup \bigsqcup_{\delta\in\Delta^{\aut{A}}} Q^{\aut{C_\delta}}_\eve$

$Q^{\Rr}_\adam\leftarrow Q_{\Ii,\adam}\sqcup \bigsqcup_{\delta\in\Delta^{\aut{A}}} Q^{\aut{C_\delta}}_\adam$

$\Delta^{\aut{R}}\leftarrow \Delta_{\Ii}\sqcup \bigsqcup_{\delta\in\Delta^{\aut{A}}} \Delta^{\aut{C_\delta}}$

\ForEach{$q\in Q^\Aa$}{
	$\Omega^\Rr(q,\star) = 0$
	
	\ForEach{$n\in N\setminus\{\star\}$} {
	\If{$\Omega^\Aa(q)\geq n$}{     
		$\Omega^\Rr(q,n) = 1$
	}	
	\Else{
		$\Omega^\Rr(q,n) = 0$
	}
}
}

\ForEach{$\delta=(q,a,q_\dL,q_\dR)\in Q^\Aa$}{
	$\Omega^\Rr(\delta,\star) = 0$
	
	\ForEach{$n\in N\setminus\{\star\}$} {
	\If{$\Omega^\Aa(q)\geq n$}{     
		$\Omega^\Rr(\delta,n) = 1$
	}	
	\Else{
		$\Omega^\Rr(\delta,n) = 0$
	}
}
}
\end{algorithm}

The idea of the automaton $\Rr$ is to simulate the following behaviour. Assume that the label of the current vertex is $a$ and the current state is $(q,n)\in Q_{\Ii,\eve}$:
\begin{itemize}
\item if $n\neq\star$ and $\Omega^\Aa(q)>n$ then \adam loses, see line~\ref{it:adam_loses};
\item \eve declares a transition $\delta=(q,a,q_\dL,q_\dR)$ of $\aut{A}$, see line~\ref{it:eve_trans};
\item \adam can decide to \emph{challenge} this transition, see line~\ref{it:adam_challenge};
\item if $n\neq\star$ then \adam chooses a direction and the game proceeds, see line~\ref{it:adam_proceed};
\item if $n=\star$ then \adam chooses a direction and a new value $n'\in N$, see line~\ref{it:adam_choose}.
\end{itemize}


Figure~\ref{fig:Rr-struct} depicts the structure of the automaton $\Rr$. The initial component $\Ii$ is split into two parts: $\Ii_0$ where $n=\star$ and $\Ii_1$ where $n\neq\star$.

We will now proceed with proving properties of the procedure \textsc{Transformation}.

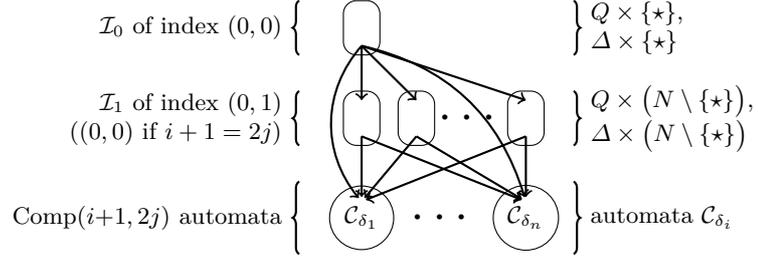
\begin{figure}
\centering
\begin{tikzpicture}[scale=0.48]
\tikzstyle{edge}=[draw, thick, ->]
\tikzstyle{aut}=[draw, circle, minimum size=24pt]

\path[draw=black, rounded corners=5] (0,0.5) rectangle ++(1,-1.5);

\foreach \x in {0, 1, 3} {
  \path[draw=black, rounded corners=5] (\x*1.5,-2) rectangle ++(1,-1.5);
}

\node[scale=2.0] at (3.5, -2.75) {$\cdots$};

\path (6.1, -1.0) edge[rbrace] node[align=left] {$Q\times\{\star\}$,\\ $\Delta\times\{\star\}$} ++(0, 1.5);

\path (-1, -1.0) edge[lbrace] node[align=right] {$\mathcal{I}_0$ of index $(0,0)$} ++(0, 1.5);

\path (6.1, -3.5) edge[rbrace] node[align=left] {$Q\times\big(N\setminus\{\star\}\big)$,\\
$\Delta\times\big(N\setminus\{\star\}\big)$} ++(0, 1.5);

\path (-1, -3.5) edge[lbrace] node[align=right] {$\mathcal{I}_1$ of index $(0,1)$\\({\small $(0,0)$ if $\rmin+1=2\rmax$})} ++(0, 1.5);

\path (6.1, -6.5) edge[rbrace] node {automata $\mathcal{C}_{\delta_i}$} ++(0, 2.0);

\path (-1, -6.5) edge[lbrace] node {$\comp(i{+}1,2j)$ automata} ++(0, 2.0);

\foreach \x in {0, 1, 3} {
  \draw[edge] (0.5, -0.75) -- (0.5+1.5*\x, -2.25);
}

\tikzstyle{mm} = [circle, inner sep=0pt, minimum size=3pt]

\node[mm] (a0) at (0.5+1.5*0, -5.05) {};
\node[mm] (a3) at (0.5+1.5*3, -5.05) {};

\foreach \x in {0, 1, 3} {
  \foreach \y in {0, 3} {
    \draw[edge] (0.5+\x*1.5, -3.25) -- (a\y);
  }
}

\node[scale=2.0] at (2.75, -5.5) {$\cdots$};

\node[aut] at (0.5, -5.5) {$\mathcal{C}_{\delta_1}$};
\node[aut] at (5.0, -5.5) {$\mathcal{C}_{\delta_n}$};

\draw (0.5, -0.75) edge[edge, bend right=40] (a0);
\draw (0.5, -0.75) edge[edge, bend right=-30] (a3);

\end{tikzpicture}
\caption{The structure of the automaton $\Rr$.}
\label{fig:Rr-struct}
\end{figure}

\begin{lemma}
If $\Aa$ is an unambiguous automaton of index $(\rmin,2\rmax)$ then $\Rr$ is a~$\comp(\rmin+1,2\rmax)$ automaton.
\end{lemma}

\begin{proof}
We first argue that if $\rmin+1<2\rmax$ then $\Rr$ is a $\comp(\rmin+1,2\rmax)$ automaton. Note every strongly-connected component in the graph of $\Rr$ is either a component of $\Ii_0$, $\Ii_1$, or of $\Cc_\delta$ for $\delta\in\Delta^\Aa$. Recall that all the components $\Aa_\delta$ are by the construction $\comp(\rmin+1,2\rmax)$-automata. By the definition, $\Ii_0$ and $\Ii_1$ are $\comp(1,2)$-automata, so the whole automaton $\aut{R}$ is also $\comp(\rmin+1,2n)$.

Consider the opposite case: $\rmin+1=2\rmax$. By shifting all the priorities we can assume that $\rmin=\rmax=1$ (i.e.~$\aut{A}$ is B\"uchi). Observe that the only possible odd value $n$ between $\rmin$ and $2\rmax$ is $n=1$. It means that if \adam declares a value $n\neq\star$ then always $\Omega(q)\geq n$ holds. It means that there are no states in $\Ii_1$ with priority $1$. Therefore, both $\Ii_0$ and $\Ii_1$ are $\comp(0,0)$ automata and $\aut{R}$ is a $\comp(0,0)$ automaton.
\end{proof}

%% file: 4_correctness.tex
\section{Correctness of the construction}
\label{sec:correctness}

In this section we prove that the automaton $\Rr$ constructed by the algorithm \textsc{Transformation} recognises the same language as the given unambiguous automaton $\Aa$. 
Let $\Aa$ be an unambiguous automaton of index $(\rmin,2\rmax)$. 

\begin{definition}
\label{def:unique}
Let $t\in\trees_A$ be a tree. 
We define $\rho_t$ as the unique maximal partial run $\rho_t$ of $\aut{A}$ on $t$, i.e.~a partial function $\parfun{\rho_t}{\{\dL,\dR\}^\ast}{Q^{\aut{A}}}$ such that:
\begin{itemize}
\item $\rho_t(\epsilon)=q_\init^\aut{A}$;
\item if $\finA\in\dom(\rho_t)$ and $t\restr_\finA\in\lang(\Cc_\delta)$ for some $\delta\in\Delta^\Aa$ then\footnote{By Lemma~\ref{lem:C-delta-properties} there is at most one such $\delta$.}
\begin{equation}
\label{eq:rho-t-constraint}
\delta=\big(\rho_t(\finA),t(\finA),\rho_t(\finA\dL),\rho_t(\finA\dR)\big); \tag{\color{blue}$\diamond$}
\end{equation}
\item if $\finA\in\dom(\rho_t)$ and $t\restr_\finA\notin\lang(\Cc_\delta)$ for any $\delta\in\Delta^\Aa$ then $\finA\dL,\finA\dR\notin\dom(\rho_t)$.
\end{itemize}
\end{definition}

\begin{lemma}
\label{lem:unique}
$t\in\lang(\Aa)$ if and only if $\rho_t$ is total and accepting.
\end{lemma}

\begin{proof}
If $\rho_t$ is accepting then it is a witness that $t\in\lang(\aut{A})$. Let $\rho$ be an accepting run of $\aut{A}$ on $t$. We inductively prove that $\rho=\rho_t$. Take a node $\finA$ of $t$ and define $q=\rho(\finA)$, $a=t(\finA)$, $q_{\dL}=\rho_t(\finA\dL)$, and $q_{\dR}=\rho_t(\finA\dR)$. Observe that $\rho$ is a witness that $(q,a)$ is productive and for $\delta=(q,a,q_\dL,q_\dR)$ we have
\[t\in L_{\delta}\subseteq\lang(\aut{C}_{\delta}).\]
Therefore, $\rho_t(\finA\dL)=\rho(\finA\dL)$ and $\rho_t(\finA\dR)=\rho(\finA\dR)$.  $\hfill\qed$
\end{proof}

\subsection{\texorpdfstring{$\lang(\aut{A})=\lang(\aut{R})$}{L(A)=L(R)}}
\label{sec:uBuchi-correctness}

\begin{lemma}
\label{lem:la_subset_lr}
If $t\in\lang(\aut{A})$ then $t\in\lang(\aut{R})$.
\end{lemma}

\begin{proof}
Assume that $t\in\lang(\Aa)$. By Lemma~\ref{lem:unique} we know that $\rho_t$ is the unique accepting run of $\Aa$ on $t$. Consider the following strategy $\sigma_\eve$ for \eve in the initial component $\Ii$ of the automaton $\Rr$: always declare $\delta$ consistent with $\rho_t$. Extend it to the winning strategies in $\aut{C}_\delta$ whenever they exist. That is, if the current vertex is $\finA$ and the state of $\aut{R}$ is of the form $(q,n)\in \Ii$ then move to the state $(\delta,n)$ for $\delta=(\rho_t(\finA),t(\finA),\rho_t(\finA \dL),\rho_t(\finA\dR))$. Whenever the game moves from the initial component $\Ii$ into one of the automata $\aut{C}_\delta$ in a vertex $\finA$, fix some winning strategy in $\game(\aut{C}_\delta, t\restr_\finA)$ (if exists) and play according to this strategy; if there is no such strategy, play using any strategy.
Take a play consistent with $\sigma_\eve$ in $\game(\aut{R},t)$. There are the following cases:
\begin{itemize}[noitemsep,topsep=0pt,parsep=0pt,partopsep=0pt]
\item \adam loses in a finite time according to the transition from line~\ref{it:adam_loses} in the algorithm \textsc{Transformation}.
\item \adam stays forever in the initial component $\Ii$ never changing the value of $n=\star$ and loses by the parity criterion.
\item In some vertex $\finA$ of the tree \adam \emph{challenges} the transition $\delta$ given by \eve and the game proceeds to the state $q_\init^{\Cc_\delta}$. In that case $t\restr_\finA\in L_\delta$ by the definition of $L_\delta$ (the run $\rho_t\restr_\finA$ is a witness) and therefore $t\restr_\finA\in\lang(\Cc_\delta)$. So \eve has a winning strategy in $\game(\Cc_\delta,t\restr_\finA)$ and \eve wins the rest of the game.
\item \adam declares a value $n\neq\star$ at some point and then never \emph{challenges} \eve. In that case the game follows an infinite branch $\infA$ of $t$. Since $\rho_t$ is accepting so we know that $k\eqdef\limsup_{i\to\infty} \Omega^\aut{A}(\rho_t(\infA\restr_i))$ is even. If $k>n$ then \adam loses at some point according to the transition from line~\ref{it:adam_loses}. Otherwise $k<n$ and from some point on all the states of $\aut{R}$ visited during the game have priority $0$, thus \adam loses by the parity criterion in $\Ii_1$.  $\hfill\qed$
\end{itemize}
\end{proof}

\begin{lemma}
If $t\notin\lang(\aut{A})$ then $t\notin\lang(\aut{R})$.
\end{lemma}

\begin{proof}
We assume that $t\notin\lang(\aut{A})$ and define a winning strategy for \adam in the game $\game(\Rr,t)$. Let us fix the run $\rho_t$ as in Definition~\ref{def:unique}.

Note that either $\rho_t$ is a partial run: there is a vertex $\finA$ such that $\rho_t(\finA)=q$ and $(q,t(\finA))$ is not productive; or $\rho_t$ is a total run. Since $t\notin\lang(\aut{A})$, $\rho_t$ cannot be a total accepting run. Let $\infA$ be a finite or infinite branch: either $\infA\in\{\dL,\dR\}^\ast$ and $\infA$ is a leaf of $\rho_t$ or $\infA$ is an infinite branch such that $k\eqdef\limsup_{i\to\infty} \Omega^\aut{A}(\rho_t(\infA\restr_i))$ is odd. If $\infA$ is finite let us put any odd value between $\rmin$ and $2\rmax$ as $k$.
Consider the following strategy for \adam:
\begin{itemize}[noitemsep,topsep=0pt,parsep=0pt,partopsep=0pt]
\item \adam keeps $n=\star$ until there are no more states of priority greater than $k$ along $\infA$ in $\rho_t$. Then he declares $n'=k$.
\item \adam \emph{challenges} a transition $\delta$ given by \eve in a vertex $\finA$ if and only if $t\restr_\finA\notin \Cc_\delta$.
\item \adam always follows $\infA$: in a vertex $\finA\in\{\dL,\dR\}^\ast$ he chooses the direction $d$ in such a way that $\finA d\preceq\infA$.
\end{itemize}

As in the proof of Lemma \ref{lem:la_subset_lr}, we extend this strategy to strategies in the components $\aut{C}_\delta$ whenever such strategies exist: if the game moves from the component $\Ii$ into one of the component $\aut{C}_\delta$ in a vertex $\finA$ then \adam uses some winning strategy in the game $\Gg(\aut{C}_\delta,t\restr_\finA)$ (if it exists); if there is no such strategy, \adam plays using any strategy.

Consider any play $\plyA$ consistent with $\sigma_\adam$. Note that if $\infA$ is a finite word and the play $\plyA$ reaches the vertex $\infA$ in a state $(\delta,n)$ in $\Ii$ then 
by the definition of $\rho_t$ we know that $t\restr_\finA\notin \Cc_\delta$ and thus \adam \emph{challenges} this transition and wins in the game $\game(\aut{C}_{\delta},t\restr_\finA)$. By the definition of the strategy $\sigma_\adam$, \adam never loses according to the transition from line~\ref{it:adam_loses}  in the algorithm \textsc{Transformation} --- if \adam declared $n\neq \star$ then the play will never reach a state of priority greater than $n$.

Let us consider the remaining cases. First assume that at some vertex $\finA$ player \adam \emph{challenged} a transition $\delta$ declared by \eve. It means that $t\restr_\finA\notin\lang(\Cc_{\delta})$ and \adam has a winning strategy in $\game(\aut{C}_{\delta},t\restr_\finA)$ and wins in that case.

The last case is that \adam did not \emph{challenge} any transition declared by \eve and the play followed the branch $\infA$. Then, for every $i\in\N$ the game reached the vertex $\infA\restr_i$ in a state $(q,n)$ satisfying $q=\rho_t(\infA\restr_i)$. In that case there is some vertex $\finA$ along $\infA$ where \adam declared $n=k$. Therefore, infinitely many times $\Omega(q)=n$ in $\plyA$ so \adam wins that play by the parity criterion. $\hfill\qed$
\end{proof}

%% file: 5_conclusion.tex
\section{Conclusion}
We presented a new algorithm \textsc{Transformation} which for a given unambiguous automaton $\Aa$ of index $(\rmin,2\rmax)$ outputs an automaton $\textsc{Transformation}(\Aa)$ which accepts the same language and belongs to the class $\comp(\rmin+1,2\rmax)$. In particular, if $\Aa$ is an unambiguous B\"uchi automaton, 
then  $\textsc{Transformation}(\Aa)$ is a weak alternating automaton. This can be considered an automata-theoretic counterpart of the  Lusin-Souslin Theorem~\cite[Theorem~15.1]{kechris_descriptive}. 

\subsection{Further work}

This paper is a part of a broader project intended to understand better the descriptive complexity of unambiguous languages of infinite trees. In our view the crucial question is whether unambiguous automata can reach arbitrarily high levels in the alternating index hierarchy. 

\smallskip
\noindent
{\em Conjecture. } There exists a pair $(\rmin,\rmax)$ such that if $\Aa$ is an unambiguous automaton on infinite trees then the language recognised by $\Aa$ can be recognised by an alternating automaton of index $(\rmin,\rmax)$.